\documentclass[runningheads,a4paper]{llncs}

\usepackage{latexsym}
\usepackage[2cell,matrix,all]{xy}
\usepackage{verbatim,alltt}
\usepackage{stmaryrd}
\usepackage{graphicx}
\usepackage{amssymb}
\usepackage{mathabx}
\usepackage{bbm}
\usepackage{amsfonts}
\usepackage{mathrsfs}
\usepackage{booktabs}
\usepackage{epsfig,amstext,graphics}
\usepackage{wasysym}
\usepackage{colortbl}
\usepackage{url}
\usepackage{wrapfig}
\usepackage[english]{babel}
\usepackage{hhline}

\newcommand{\sra}{\mbox{{\footnotesize $\rightarrow$}}}

\begin{document}

\mainmatter  


\title{Variant-Based Decidable Satisfiability in Initial Algebras with Predicates\thanks{Partially supported by NSF Grant CNS
    14-09416, the EU (FEDER), Spanish MINECO project
    TIN2015-69175-C4-1-R and GV project PROMETEOII/2015/013.}}

\titlerunning{Variant Satisfiability in Initial Algebras with Predicates}

\author{Ra\'ul Guti\'errez\inst{1} and Jos\'{e} Meseguer\inst{2}}

\authorrunning{R. Gutierrez and J. Meseguer}

\institute{Universitat Polit\`ecnica de Val\`encia \and University of Illinois at Urbana-Champaign}

\maketitle

\begin{abstract} Decision procedures can be either \emph{theory-specific},
  e.g., Presburger arithmetic, or \emph{theory-generic}, applying to
  an infinite number of user-definable theories.  Variant
  satisfiability is a theory-generic procedure for quantifier-free
  satisfiability in the initial algebra of an  order-sorted  equational theory
  $(\Sigma,E \cup B)$ under two conditions: (i) $E \cup B$ has the 
\emph{finite variant property} and $B$ has a finitary unification
algorithm; and (ii) $(\Sigma,E \cup B)$ protects a constructor
subtheory  $(\Omega,E_{\Omega} \cup B_{\Omega})$ that is
OS-\emph{compact}.  These conditions apply to many
user-definable theories, but have a main limitation: they apply well to
\emph{data structures}, but often do \emph{not} hold for
user-definable \emph{predicates} on such data structures.
We present a theory-generic satisfiability decision
procedure, and a prototype implementation, extending variant-based 
satisfiability  to initial algebras with user-definable predicates under
fairly general conditions.

\end{abstract}

\section{Introduction}

Some of the most important recent advances in software verification
are due to the systematic use of decision procedures in both model
checkers and theorem provers.  However, a key limitation in exploiting
the power of such decision procedures is their current \emph{lack of
  extensibility}.  The present situation is as follows.  Suppose a
system has been formally specified as a theory $T$ about which we want
to verify some properties, say $\varphi_{1},\ldots,\varphi_{n}$, using
some model checker or theorem prover that relies on an SMT solver for
its decision procedures.  This limits \emph{a priori} the decidable
subtheory $T_{0} \subseteq T$ that can be handled by the SMT solver.
Specifically, the SMT solver will typically support a fixed set
$Q_{1},\ldots,Q_{k}$ of decidable theories, so that, using a theory
combination method such as Nelson and Oppen
\cite{DBLP:journals/toplas/NelsonO79}, or Shostak
\cite{Shostak:combination}, $T_{0}$ must be a finite combination of
the decidable theories $Q_{1},\ldots,Q_{k}$ supported by the SMT
solver.

In non-toy applications it is unrealistic to expect that the entire
specification $T$ of a software system will be decidable.  Obviously,
the bigger the decidable subtheory  $T_{0} \subseteq T$, the
higher the levels of automation and the greater the chances of scaling
up the verification effort.  With
\emph{theory-specific} procedures for, say, $Q_{1},\ldots,Q_{k}$,
the decidable fragment $T_{0}$ of $T$ is a
priori bounded.  One promising way to extend the decidable fragment
$T_{0}$ is to develop \emph{theory-generic} satisfiability
procedures.  These are procedures that 
make decidable not a single theory $Q$, but an \emph{infinite class}
of  \emph{user-specifiable} theories.  Therefore,
an SMT solver supporting both theory-specific and theory-generic
decision procedures becomes \emph{user-extensible} and can carve out a
potentially much bigger 
Decidable Fragment $T_{0}$ of the given system specification $T$.

Variant-based satisfiability \cite{var-sat,var-sat-short} is a recent
theory-generic decision procedure applying to the following, easily
user-specifiable infinite class of equational theories $(\Sigma,E \cup
B)$: (i) $\Sigma$ is an order-sorted \cite{osa1} signature of function
symbols, supporting types, subtypes, and subtype polymorphisms; (ii) $E \cup B$ has the \emph{finite variant property} \cite{comon-delaune}
and $B$ has a finitary unification algorithm; and (iii) $(\Sigma,E
\cup B)$ protects a constructor subtheory $(\Omega,E_{\Omega} \cup
B_{\Omega})$ that is OS-\emph{compact} \cite{var-sat,var-sat-short}.
The procedure can then decide satisfiability in the initial algebra
$T_{\Sigma/E \cup B}$, that is, in the \emph{algebraic data type}
specified by $(\Sigma,E \cup B)$.  These conditions apply to many
user-definable theories, but have a main limitation: they apply well
to \emph{data structures}, but often do \emph{not} hold for
user-definable \emph{predicates}.

The notions of variant and of OS-compactness mentioned above are
defined in detail in Section \ref{osa-prelims}.  Here we give some key
intuitions about each notion.  Given $\Sigma$-equations $E \cup
B$ such that
the equations $E$ oriented
as left-to-right rewrite rules are confluent and terminating
modulo the equational  axioms $B$,  a \emph{variant} of a $\Sigma$-term $t$ is
a pair $(u,\theta)$ where $\theta$ is a substitution, and $u$ is the
canonical form of the term instance $t \theta$ by the rewrite rules
$E$ modulo $B$.  Intuitively, the variants of $t$ are the fully
simplified \emph{patterns} to which the instances of $t$ can reduce.
Some simplified instances are of course more general (as patterns) than others.
 $E \cup B$ has the \emph{finite variant property} (FVP) if any
 $\Sigma$-term $t$ has a \emph{finite} set of most general variants.
For example, the addition equations $E=\{x+0=x,x+s(y)=s(x+y)\}$
are \emph{not} FVP, since $(x+y,\mathit{id})$, $(s(x+y_{1}),
\{y \mapsto s(y_{1})\})$, $(s(s(x+y_{2})),\{y \mapsto s(s(y_{2}))\})$, $\ldots$,
$(s^{n}(x+y_{n}),\{y \mapsto s^{n}(y_{n})\})$, $\ldots$, are all 
\emph{incomparable} variants of $x+y$.
Instead, the Boolean equations $G=\{x\vee \top = \top,x\vee \bot = x,x\wedge \top = x,x\wedge \bot = \bot\}$
\emph{are} FVP.  For example, the most general variants of $x \vee y$ are: 
$(x \vee y,\mathit{id})$, $(x,\{y \mapsto \bot\})$, and $(\top,\{y
\mapsto \top\})$.
Assuming for simplicity that 
all sorts in a theory $(\Omega,E_{\Omega} \cup
B_{\Omega})$ have an infinite number of ground terms of that sort
which are
all different modulo the equations $E_{\Omega} \cup
B_{\Omega}$, then OS-compactness of $(\Omega,E_{\Omega} \cup
B_{\Omega})$ means that any conjunction of disequalities
$\bigwedge_{1 \leq i \leq n} u_{i} \not= v_{i}$ such that
$E_{\Omega} \cup
B_{\Omega} \not\vdash u_{i} =v_{i}$, $1 \leq i \leq n$,
is \emph{satisfiable} in the initial
algebra $T_{\Omega/E_{\Omega} \cup B_{\Omega}}$.  For example, $(\{0,s\},\emptyset)$ is
OS-compact, where $\{0,s\}$ are the usual natural number constructors.
Thus, $s(x) \not= s(y) \wedge 0 \not = y$ \emph{is} satisfiable
in $T_{\{0,s\}}$.

The key reason why user-definable predicates present a serious
obstacle is the following.  Variant satisfiability works by
\emph{reducing} satisfiability in the initial algebra
$T_{\Sigma/E \cup B}$ to satisfiability in the much simpler algebra of
constructors $T_{\Omega/E_{\Omega} \cup B_{\Omega}}$.  In many
applications $E_{\Omega} = \emptyset$, and if the axioms $B_{\Omega}$
are any combination of associativity, commutativity and identity
axioms, except associativity without commutativity, then
$(\Omega,B_{\Omega})$ is an OS-compact theory
\cite{var-sat,var-sat-short}, making satisfiability in
$T_{\Omega/ B_{\Omega}}$ and therefore in $T_{\Sigma/E \cup B}$
decidable.  We can equationally specify a predicate $p$ with sorts
$A_{1},\ldots,A_{n}$ in a \emph{positive} way as a function
$p:A_{1},\ldots,A_{n} \rightarrow \mathit{Pred}$, where the sort
$\mathit{Pred}$ of predicates contains a ``true'' constant
$\mathit{tt}$, so that $p(u_{1},\ldots,u_{n})$ not holding for
concrete ground arguments $u_{1},\ldots,u_{n}$ is expressed as the
\emph{disequality} $p(u_{1},\ldots,u_{n}) \not= \mathit{tt}$.  But
$p(u_{1},\ldots,u_{n}) \not= \mathit{tt}$ means that $p$ must be a
\emph{constructor} of sort $\mathit{Pred}$ in $\Omega$, and that the
equations defining $p$ must belong to $E_{\Omega}$, making
$E_{\Omega} \not= \emptyset$ and ruling out the case when
$T_{\Omega/E_{\Omega} \cup B_{\Omega}} = T_{\Omega/ B_{\Omega}}$ is
decidable by OS-compactness.

This work extends variant-based 
satisfiability  to initial algebras with user-definable predicates under
fairly general conditions using two key ideas: (i)
characterizing the cases when $p(u_{1},\ldots,u_{n}) \not=
\mathit{tt}$ by means of \emph{constrained patterns};
and (ii) eliminating all occurrences of disequalities
of the form $p(v_{1},\ldots,v_{n}) \not=
\mathit{tt}$ in a quantifier-free (QF) formula by means of such 
patterns.  In this way, the QF satisfiability problem
can be reduced to formulas involving only
\emph{non-predicate} constructors, for which OS-compactness 
holds in many applications.  More generally, if some 
predicates fall within the OS-compact fragment, they can be kept.

Preliminaries are in Section \ref{osa-prelims}. Constructor variants and OS-compactness
 in Section \ref{gvars-gvunif}. 
The satisfiability decision procedure is
defined and proved correct in Section
\ref{inductive-validity-with-preds}, and its prototype
implementation is described in Section \ref{implementation}.
Related work and conclusions are discussed in Section
\ref{related-concl}.  All proofs can be found in \cite{gutierrez-meseguer-var-pred-tech-rep}.

\section{Many-Sorted Logic, Rewriting, and Variants} 
\label{osa-prelims}

	We present some preliminaries on many-sorted (MS) logic,
        rewriting and finite variant and variant unification notions needed in the paper.
        For a more general treatment using order-sorted (OS) logic see
        \cite{gutierrez-meseguer-var-pred-tech-rep}.

        We assume familiarity with the following basic concepts and
        notation that are explained in full detail in, e.g.,
        \cite{mg85}: (i) {\em many-sorted (MS) signature\/} as a pair
        $\Sigma = (S,\Sigma)$ with $S$ a set of \emph{sorts} and
        $\Sigma$ an $S^{*} \times S$-indexed family
        $\Sigma = \{\Sigma_{w,s}\}_{(w,s) \in S^{*} \times S}$ of
        function symbols, where $f \in \Sigma_{s_{1} \ldots s_{n},s}$
        is displayed as $f: s_{1} \ldots s_{n} \rightarrow s$; (ii)
        {\em $\Sigma$-algebra\/} $A$ as a pair $A =(A,\_{_{A}})$ with
        $A=\{A_{s}\}_{s \in S}$ an $S$-indexed family of sets, and
        $\_{_{A}}$ a mapping interpreting each
        $f: s_{1} \ldots s_{n} \rightarrow s$ as a function in the set
        $[A_{s_{1}} \times \ldots \times A_{s_{n}}\rightarrow
        A_s]$. (iii) \emph{$\Sigma$-homomorphism\/}
        $h: A \rightarrow B$ as an $S$-indexed family of functions
        $h =\{h_{s}: A_{s} \rightarrow B_{s}\}_{s \in S}$ preserving
        the operations in $\Sigma$; (iv) the term $\Sigma$-algebra
        $T_{\Sigma}$ and its initiality in the category
        ${\bf MSAlg}_{\Sigma}$ of $\Sigma$-algebras when $\Sigma$ is
        unambiguous.
	
	An $S$-sorted set $X=\{X_{s}\}_{s \in S}$ of \emph{variables},
	satisfies $s \not= s' \Rightarrow X_{s}\cap
	X_{s'}=\emptyset$, and the variables in $X$ are always assumed
	\emph{disjoint} from all constants in $\Sigma$.  The $\Sigma$-\emph{term
		algebra} on variables $X$,
	$T_{\Sigma}(X)$, is the \emph{initial algebra} for the signature
	$\Sigma(X)$ obtained by adding to $\Sigma$ the variables $X$ \emph{as
		extra constants}.  Since a $\Sigma(X)$-algebra is just a pair
	$(A,\alpha)$, with $A$ a $\Sigma$-algebra, and $\alpha$ an
	\emph{interpretation of the constants} in $X$, i.e.,
	an $S$-sorted function $\alpha \in [X \sra A]$, the
	$\Sigma(X)$-initiality of
	$T_{\Sigma}(X)$ means that 
		for each $A \in {\bf MSAlg}_{\Sigma}$ and
		$\alpha \in [X \sra A]$, there exists 
		a unique $\Sigma$-homomorphism, 
		$\_\alpha : T_{\Sigma}(X) \rightarrow A$ extending $\alpha$, i.e.,
		such that for each $s \in S$ and $x \in X_{s}$ we have $x
		\alpha_{s} = \alpha_{s}(x)$.
	 In particular, when $A=T_{\Sigma}(Y)$, an interpretation of
	the constants in $X$, i.e., 
	an $S$-sorted function $\sigma\in[X \sra T_{\Sigma}(Y)]$ is called
	a \emph{substitution}, 
	and its unique homomorphic extension $\_\sigma : T_{\Sigma}(X)
	\rightarrow T_{\Sigma}(Y)$ is also called a substitution.
	Define $\mathit{dom}(\sigma)=\{x \in X \mid x \not= x \sigma\}$,
	and $\mathit{ran}(\sigma)= \bigcup_{x \in \mathit{dom}(\sigma)}
	\mathit{vars}(x \sigma)$. Given variables $Z$, the substitution
	$\sigma|_{Z}$ agrees with $\sigma$ on $Z$ and
	is the identity elsewhere.

       We also assume familiarity with many-sorted first-order logic including: (i) the  first-order language of 
	$\Sigma$-\emph{formulas} for $\Sigma$ a signature (in our case
        $\Sigma$ has only function symbols and the $=$ predicate);
        (ii) given a $\Sigma$-algebra $A$, a formula $\varphi \in \mathit{Form}(\Sigma)$,
	and an assignment $\alpha \in [Y \sra A]$, with
	$Y=\mathit{fvars}(\varphi)$ the free variables of $\varphi$,
	the \emph{satisfaction
		relation} $A,\alpha \models \varphi$;
       (iii) the notions of a formula $\varphi \in
       \mathit{Form}(\Sigma)$ being \emph{valid}, denoted $A \models
       \varphi$, resp. \emph{satisfiable}, in a $\Sigma$-algebra $A$.
      For a subsignature $\Omega \subseteq
	\Sigma$ and $A \in {\bf
		MSAlg}_{\Sigma}$, the \emph{reduct} $A|_{\Omega} \in {\bf
		MSAlg}_{\Omega}$ agrees with
	$A$ in the interpretation of all sorts and operations in $\Omega$ and
	discards everything in $\Sigma \setminus \Omega$.
	If $\varphi \in \mathit{Form}(\Omega)$ we have
	the equivalence $A \models \varphi \; \Leftrightarrow \;
        A|_{\Omega} \models \varphi$.

	An MS \emph{equational theory}
	is a pair $T=(\Sigma,E)$, with $E$ a set of  $\Sigma$-equations.
	${\bf MSAlg}_{(\Sigma,E)}$ denotes the full subcategory
	of ${\bf MSAlg}_{\Sigma}$ with objects those $A \in {\bf
		MSAlg}_{\Sigma}$ such that $A \models E$, called the $(\Sigma,E)$-\emph{algebras}.
	${\bf MSAlg}_{(\Sigma,E)}$ has an
	\emph{initial algebra} $T_{\Sigma/E}$ \cite{mg85}.
	The inference system in \cite{mg85} is
	\emph{sound and complete} for MS  equational deduction, i.e., for
	any MS equational theory
	$(\Sigma,E)$, and $\Sigma$-equation $u=v$
	we have an equivalence $E \vdash u=v \; \Leftrightarrow \; E
	\models u=v$. For the sake of simpler inference we assume
        \emph{non-empty sorts}, i.e., $\forall s \in S,\; T_{\Sigma},s \not= \emptyset$.
       Deducibility $E \vdash u=v$ is abbreviated 
	as $u =_{E} v$.

In the above notions there is only an \emph{apparent} lack of
predicate symbols: full many-sorted first-order logic can be \emph{reduced}
to many-sorted algebra and the above language of equational
formulas. 
The reduction is achieved as follows. 
A many-sorted first-order (MS-FO) signature, is a pair $(\Sigma,\Pi)$
with $\Sigma$ a MS signature
with set of sorts $S$, and $\Pi$ an $S^{\ast}$-indexed set 
$\Pi=\{\Pi_{w}\}_{w \in S^{\ast}}$ of \emph{predicate symbols}. 
 We associate to
a MS-FO signature  $(\Sigma,\Pi)$ a MS  signature $(\Sigma \cup \Pi)$
by adding to $\Sigma$ a new sort
$\mathit{Pred}$ with a constant $\mathit{tt}$ and viewing each $p \in \Pi_{w}$ as a function
symbol $p:s_{1}\ldots s_{n} \rightarrow \mathit{Pred}$.
The reduction at the model level is now very simple:
each $(\Sigma \cup \Pi)$-algebra $A$ defines a $(\Sigma,\Pi)$-model
$A^{\circ}$ with $\Sigma$-algebra structure $A|_{\Sigma}$
and having for each  $p \in \Pi_{w}$ the predicate interpretation
$A^{\circ}_{p} =A^{-1}_{p:w \rightarrow \mathit{Pred}}(\mathit{tt})$.
The reduction at the formula level is also quite simple:
we map a $(\Sigma,\Pi)$-formula $\varphi$ to an 
equational formula $\widetilde{\varphi}$, called its \emph{equational version},
by just replacing each atom $p(t_{1},\ldots,t_{n})$
by the equational atom $p(t_{1},\ldots,t_{n}) = \mathit{tt}$.
The \emph{correctness} of this reduction is just the
easy to check equivalence:
\[A^{\circ} \models \varphi \; \Leftrightarrow \; A \models
\widetilde{\varphi}.\]
A MS-FO \emph{theory} is just a pair $((\Sigma,\Pi),\Gamma)$,
with $(\Sigma,\Pi)$ a MS-FO signature and
$\Gamma$ a set of $(\Sigma,\Pi)$-formulas.
Call $((\Sigma,\Pi),\Gamma)$ \emph{equational}
iff $(\Sigma \cup \Pi,\widetilde{\Gamma})$  is a many-sorted equational theory.
By the above equivalence and the completeness of many-sorted equational logic
such theories allow a sound and complete use of equational deduction also
with predicate atoms. Note that if  $((\Sigma,\Pi),\Gamma)$
is equational, it is a very simple type of theory in many-sorted Horn
Logic with Equality and therefore has an initial
model $T_{(\Sigma,\Pi),\Gamma}$ \cite{tapsoft87}.
A useful, easy to check fact is that  we have 
an identity: $T^{\circ}_{\Sigma \cup \Pi/ \widetilde \Gamma} =
T_{(\Sigma,\Pi),\Gamma}$.

Recall the notation for
term positions, subterms, and term replacement
from \cite{dershowitz-jouannaud}: (i) positions in a term viewed as a tree are
marked by strings $p \in \mathbb{N}^{*}$ specifying
a path from the root, (ii)  $t|_{p}$ denotes the
subterm of term $t$ at position $p$, and (iii) $t[u]_{p}$ denotes the result
of \emph{replacing} subterm $t|_{p}$ at position $p$ by $u$.

\begin{definition}  A \emph{rewrite theory} is a triple
$\mathcal{R}=(\Sigma,B,R)$ with $(\Sigma,B)$ a MS equational
theory and $R$ a set of $\Sigma$-\emph{rewrite rules}, i.e., 
sequents $l \rightarrow r$, with
$l,r \in T_{\Sigma}(X)_{s}$ for some $s \in S$.
In what follows it is always assumed that:
(1)
 For each $l \rightarrow r
\in R$, $l \not\in X$ and $\mathit{vars}(r) \subseteq
\mathit{vars}(l)$.
(2)  Each equation $u=v \in B$ is \emph{regular}, i.e., $\mathit{vars}(u) =
\mathit{vars}(v)$, and \emph{linear}, i.e., there are no repeated variables 
in either $u$ or $v$.
The \emph{one-step}
  $R,B$-{rewrite relation} $t \rightarrow_{R,B} t'$, holds between
  $t,t'\in T_{\Sigma}(X)_{s}$, $s\in S$, iff there is a
  rewrite rule $l\rightarrow r \in R$, a substitution
  $\sigma \in [X \sra T_{\Sigma}(X)]$, and a term position $p$ in $t$
  such that $t|_{p}=_{B} l \sigma$, and $t'=t[r \sigma]_{p}$.

$\mathcal{R}$ is called: (i) \emph{terminating} iff the relation
$\rightarrow_{R,B}$ is well-founded; (ii) \emph{strictly}
$B$-\emph{coherent} \cite{DBLP:journals/tcs/Meseguer17}
iff whenever $u \rightarrow_{R,B} v$ and $u
=_{B}u'$ there is a $v'$ such that 
$u' \rightarrow_{R,B} v'$ and $v =_{B}v'$;
 (iii) \emph{confluent}
iff $u \rightarrow^{*}_{R,B} v_1$ and $u \rightarrow^{*}_{R,B} v_2$
imply that there are $w_{1},w_{2}$ such that
$v_1 \rightarrow^{*}_{R,B} w_1$, $v_2 \rightarrow^{*}_{R,B} w_2$,
and $w_{1} =_{B}w_{2}$ (where $\rightarrow^{*}_{R,B}$ denotes  the
reflexive-transitive closure of $\rightarrow_{R,B}$); and (iv)
\emph{convergent} if (i)--(iii) hold.
If $\mathcal{R}$ is convergent, for each $\Sigma$-term $t$ 
there is a term $u$ such that $t \rightarrow^{*}_{R,B}u$
and $(\not\exists v)\; u \rightarrow_{R,B}v$.   We then
write $u=t!_{R,B}$ and
$t \rightarrow!_{R,B} t!_{R,B}$, and call $t!_{R,B}$ the
$R,B$-\emph{normal form} of $t$, which, by confluence, 
is unique up to $B$-equality.
\end{definition}

Given a set $E$ of $\Sigma$-equations, let
$R(E)=\{u \rightarrow v \mid u=v \in E\}$. A \emph{decomposition} of
a  MS equational theory $(\Sigma,E)$ is a convergent rewrite
theory $\mathcal{R}=(\Sigma,B,R)$ such that $E = E_{0} \uplus B$ and $R=R(E_{0})$.
The key property of a decomposition is the following:

\begin{theorem} (Church-Rosser Theorem)
  \cite{jouannaud-hkirchner,DBLP:journals/tcs/Meseguer17} 
Let $\mathcal{R}=(\Sigma,B,R)$ be a decomposition of $(\Sigma,E)$.
Then we have an equivalence:
\[E \vdash u = v \;\; \Leftrightarrow \;\;  u!_{R,B} =_{B} v!_{R,B}.\]
\end{theorem}

If $\mathcal{R}=(\Sigma,B,R)$ is a decomposition of $(\Sigma,E)$, 
and $X$ an $S$-sorted set of variables, the
\emph{canonical term algebra} $C_{\mathcal{R}}(X)$ has
$C_{\mathcal{R}}(X)_{s}= \{ [t!_{R,B}]_{B} \mid t \in
T_{\Sigma}(X)_{s}\}$,
and interprets each $f:s_{1} \ldots s_{n} \rightarrow s$ as the
function
$C_{\mathcal{R}}(X)_{f}: ([u_1]_{B}, \ldots, [u_n]_{B}) \mapsto [f(u_1,
\ldots,u_n)!_{R,B}]_{B}$.  By the Church-Rosser Theorem we then
have an isomorphism $h:T_{\Sigma/E}(X) \cong C_{\mathcal{R}}(X)$, where
$h: [t]_{E} \mapsto [t!_{R,B}]_{B}$.  
In particular, when $X$ is the empty family of variables, the
canonical term algebra $C_{\mathcal{R}}$ is an initial algebra, and
 is the most intuitive possible model for $T_{\Sigma/E}$ as an algebra of
\emph{values} computed by $R,B$-simplification. 

Quite often, the signature $\Sigma$ on which $T_{\Sigma/E}$ is defined
has a natural decomposition as a
disjoint union $\Sigma = \Omega \uplus \Delta$, where
the elements of $C_{\mathcal{R}}$, that is, the 
\emph{values} computed by $R,B$-simplification, are $\Omega$-terms,
whereas the function symbols $f \in  \Delta$ are viewed as
\emph{defined functions} which are
\emph{evaluated away} by $R,B$-simplification.
$\Omega$ (with same poset of sorts as $\Sigma$)
is then called a \emph{constructor subsignature} of $\Sigma$.
Call a decomposition $\mathcal{R}=(\Sigma,B,R)$
of   $(\Sigma,E)$ 
\emph{sufficiently complete} with
respect to the \emph{constructor subsignature} $\Omega$
iff  for each $t \in T_{\Sigma}$ we
have: (i) $t!_{R,B} \in T_{\Omega}$, and (ii) if $u \in T_{\Omega}$
and $u =_{B} v$, then $v \in T_{\Omega}$.  This
ensures that for each $[u]_{B} \in C_{\mathcal{R}}$
we have $[u]_{B} \subseteq T_{\Omega}$.
%
%
We will give several examples of
decompositions  $\Sigma = \Omega \uplus \Delta$ into constructors
and defined functions.

As we can see in the following definition, sufficient completeness is
closely related to the notion of a \emph{protecting} theory inclusion.
%

\begin{definition}  \label{protecting}
 An equational
theory $(\Sigma,E)$ 
\emph{protects} 
another theory $(\Omega,E_{\Omega})$  iff
$(\Omega,E_{\Omega})\subseteq (\Sigma,E)$ 
and the unique $\Omega$-homomorphism
$h:T_{\Omega / E_{\Omega}} \rightarrow T_{\Sigma/E} |_{\Omega}$ is
an isomorphism $h:T_{\Omega / E_{\Omega}} \cong T_{\Sigma/E}
|_{\Omega}$.
A decomposition $\mathcal{R}=(\Sigma,B,R)$ 
\emph{protects} 
another decomposition
$\mathcal{R}_{0}=(\Sigma_{0},B_{0},R_{0})$ 
iff $\mathcal{R}_{0}\subseteq \mathcal{R}$, i.e., $\Sigma_{0}
\subseteq \Sigma$, $B_{0} \subseteq B$, and $R_{0} \subseteq R$, and
for all  $t,t' \in T_{\Sigma_{0}}(X)$ we have:
(i) $t=_{B_0} t' \Leftrightarrow t=_{B} t'$, (ii) $t=t!_{R_{0},B_{0}}
\Leftrightarrow t=t!_{R,B}$, and (iii)
$C_{\mathcal{R}_{0}}=C_{\mathcal{R}}|_{\Sigma_{0}}$.

$\mathcal{R}_{\Omega}=(\Omega,B_{\Omega},R_{\Omega})$ is a
\emph{constructor decomposition} of $\mathcal{R}=(\Sigma,B,R)$ iff
$\mathcal{R}$ protects $\mathcal{R}_{\Omega}$ and $\Sigma$ and
$\Omega$ have the same poset of sorts, so that by (iii) above
$\mathcal{R}$ is sufficiently complete with respect to
$\Omega$.  Furthermore, $\Omega$ is called a subsignature of \emph{free
  constructors modulo} $B_{\Omega}$ iff $R_{\Omega} = \emptyset$, so
that $C_{\mathcal{R}_{\Omega}}=T_{\Omega/B_{\Omega}}$.
\end{definition}

The case where all constructor terms are in $R,B$-normal
form is captured by $\Omega$ being a subsignature
of free constructors modulo $B_{\Omega}$.  Note also that conditions
(i) and (ii) are, so called, ``no confusion'' conditions, and
for protecting extensions
(iii) is a ``no junk'' condition, that is, $\mathcal{R}$ does not
add new data to $C_{\mathcal{R}_{0}}$.

Given a MS equational theory $(\Sigma,E)$ and a conjunction 
of
$\Sigma$-equations 
$\phi=u_{1}=v_{1} \, \wedge \, \ldots \, \wedge \,  u_{n}=v_{n}$, an $E$-\emph{unifier} of $\phi$
is a substitution $\sigma$ such that $u_{i} \sigma=_{E} v_{i}\sigma$,
$1 \leq i \leq n$.  An $E$-\emph{unification algorithm} for
$(\Sigma,E)$ 
is an algorithm generating for each system of $\Sigma$-equations $\phi$
and finite set of variables
 $W \supseteq \mathit{vars}(\phi)$ a \emph{complete set} of $E$-unifiers
$\mathit{Unif}^{W}_{E}(\phi)$ where each $\tau \in
\mathit{Unif}^{W}_{E}(\phi)$ is assumed idempotent and with
$\mathit{dom}(\tau)=\mathit{vars}(\phi)$, and is
``away from $W$'' in the sense that $ran(\tau) \cap W = \emptyset$.
The set $\mathit{Unif}^{W}_{E}(\phi)$ is called ``complete'' in
the precise sense that for any $E$-unifier $\sigma$ of $\phi$ there
is a $\tau \in \mathit{Unif}_{E}(\phi)$ and a substitution
$\rho$ such that $\sigma|_{W} =_{E} (\tau
\rho)|_{W}$, where, by definition, $\alpha=_{E}\beta$ means
$(\forall x \in X)\; \alpha(x) =_{E} \beta(x)$ for 
substitutions $\alpha,\beta$.
Such an algorithm is called \emph{finitary} if it  always terminates
with a \emph{finite set} $\mathit{Unif}^{W}_{E}(\phi)$ for any $\phi$.

The notion of \emph{variant} answers, in a sense, 
two questions: (i) how can we best describe symbolically
the elements of $C_{\mathcal{R}}(X)$ that are \emph{reduced
substitution instances} of a
\emph{pattern term} $t$? and (ii) given an original pattern $t$, how many
other patterns do we need to describe the reduced instances of $t$ in
$C_{\mathcal{R}}(X)$?

\begin{definition} \label{variant-defn}
Given a decomposition $\mathcal{R}=(\Sigma,B,R)$ of a MS
equational theory $(\Sigma,E)$ and a $\Sigma$-term $t$, a
\emph{variant}\footnote{For a discussion of similar
  but not exactly equivalent versions of the variant notion see
  \cite{variants-of-variants}.  Here we follow the shaper formulation in
  \cite{variant-JLAP}, rather than the one in \cite{comon-delaune},
  because it is technically essential for some results to hold \cite{variants-of-variants}.}  
\cite{comon-delaune,variant-JLAP}  of $t$ is a pair
$(u,\theta)$ such that: (i) $u =_{B} (t \theta)!_{R,B}$, (ii)
$\mathit{dom}(\theta) \subseteq \mathit{vars}(t)$, and (iii)
$\theta = \theta !_{R,B}$, that is,
$\theta(x) = \theta(x) !_{R,B}$ for all variables $x$.
$(u,\theta)$ is called a \emph{ground variant} iff, furthermore, $u
\in T_{\Sigma}$.  
%
%
Given variants $(u,\theta)$ and $(v,\gamma)$ of $t$, $(u,\theta)$ is
called \emph{more general} than $(v,\gamma)$, denoted
$(u,\theta) \sqsupseteq_{B} (v,\gamma)$, iff there is a substitution
$\rho$ such that: (i) $(\theta \rho)|_{\mathit{vars}(t)} =_{B} \gamma$, and (ii)
$u \rho =_{B} v$.  Let
$\llbracket t \rrbracket_{R,B}=\{(u_{i},\theta_{i}) \mid i \in I\}$
denote a \emph{complete set of variants} of $t$, that is,
a set of variants such that for any variant $(v,\gamma)$ of $t$
there is an $i \in I$, such that
$(u_{i},\theta_{i}) \sqsupseteq_{B} (v,\gamma)$.

A decomposition $\mathcal{R}=(\Sigma,B,R)$ of $(\Sigma,E)$ has the
\emph{finite variant property} \cite{comon-delaune} (FVP) iff for each
$\Sigma$-term $t$ there is a \emph{finite} complete set
of variants
$\llbracket t \rrbracket_{R,B}=\{(u_{1},\theta_{1}),
\ldots,(u_{n},\theta_{n})\}$.  If $B$ has a finitary
$B$-unification algorithm the relation $(u,\alpha) \sqsupseteq_{B}
(v,\beta)$ is decidable by $B$-matching.  Under this assumption on $B$,
if  $\mathcal{R}=(\Sigma,B,R)$  is FVP,
$\llbracket t \rrbracket_{R,B}$ can be chosen to be not only
complete, but also a set of \emph{most general} variants,
in the sense that  for
$1 \leq i < j \leq n$,
$(u_i,\theta_i) \not\sqsupseteq_{B}
(u_j,\theta_j) \; \wedge \; (u_j,\theta_j) \not\sqsupseteq_{B}
(u_i,\theta_i)$. Also, given any finite set of variables
$W \supseteq \mathit{vars}(t)$ we can always choose
$\llbracket t \rrbracket_{R,B}$ to be of the form $\llbracket t
\rrbracket^{W}_{R,B}$, where each $(u_{i},\theta_{i}) \in \llbracket t
\rrbracket^{W}_{R,B}$ has $\theta_{i}$ idempotent with
$\mathit{dom}(\theta_{i}) = \mathit{vars}(t)$, and ``away from $W$,''
in the sense that $\mathit{ran}(\theta_{i}) \cap W = \emptyset$.
%
\end{definition}

If $B$ has a finitary unification algorithm, the \emph{folding variant
  narrowing} strategy described in~\cite{variant-JLAP} provides an
effective method to generate $\llbracket t \rrbracket_{R,B}$.
Furthermore, folding variant narrowing \emph{terminates} for each
input $t \in T_{\Sigma}(X)$ with a finite set
$\llbracket t \rrbracket_{R,B}$ iff $\mathcal{R}$ has FVP
\cite{variant-JLAP}.

Two example theories, one FVP and another not FVP, were
given in the Introduction.  Many other examples are given in \cite{var-sat}.
 The following will be used as a running example of an FVP theory:

\begin{example} \label{nat-set-fvp-example}
(Sets of Natural Numbers). Let $\mathit{NatSet}=(\Sigma,B,R)$ be the
following equational theory.  $\Sigma$ has sorts $\mathit{Nat}$,
$\mathit{NatSet}$ and $\mathit{Pred}$, subsort inclusion\footnote{As pointed out at the beginning of
          Section \ref{osa-prelims}, \cite{gutierrez-meseguer-var-pred-tech-rep}  treats the
          more general \emph{order-sorted} case, where sorts form
          a poset $(S,\leq)$ with $s \leq s'$ interpreted as set
          containment $A_{s} \subseteq A_{s'}$ in a $\Sigma$-algebra
          $A$. All results in this paper hold in the order-sorted case.}
$\mathit{Nat} < \mathit{NatSet}$, and
decomposes as
$\Sigma = \Omega_{c} \uplus \Delta$,
where the constructors $\Omega_{c}$ include the following operators:
$0$ and $1$ of sort $\mathit{Nat}$,
$\_+\_: \mathit{Nat}\,\mathit{Nat} \rightarrow \mathit{Nat}$ (addition),
$\emptyset$ of sort $\mathit{NatSet}$,
$\_,\_: \mathit{NatSet}\,\mathit{NatSet} \rightarrow \mathit{NatSet}$
(set union),
$tt$ of sort $\mathit{Pred}$, and a subset containment
predicate expressed as a function 
$\_\subseteq \_ : \mathit{NatSet}\,\mathit{NatSet} \rightarrow
\mathit{Pred}$. $B$ decomposes as $B =B_{\Omega_{c}} \uplus B_{\Delta}$.
The axioms $B_{\Omega_{c}}$ include: (i) the associativity and
commutativity of $\_+\_$ with identity $0$, the associativity and
commutativity of $\_,\_$.  $R$ decomposes as
$R =R_{\Omega_{c}} \uplus R_{\Delta}$.
The rules $R_{\Omega_{c}}$ include:
(i) an identity rule for union $NS,\emptyset \rightarrow NS$;
(ii) idempotency rules for union $NS,NS \rightarrow NS$, and
$NS,NS,NS' \rightarrow NS,NS'$; and (iii) rules
defining the $\_\subseteq \_$ predicate, 
$\emptyset \subseteq NS \rightarrow tt$,
$NS \subseteq NS \rightarrow tt$, and
$NS \subseteq NS,NS' \rightarrow tt$, 
where $NS$ and $NS'$ have sort $\mathit{NatSet}$.
The signature $\Delta$ of defined functions 
has operators $\mathit{max}: \mathit{Nat}\,\mathit{Nat} \rightarrow \mathit{Nat}$,
$\mathit{min}: \mathit{Nat}\,\mathit{Nat} \rightarrow \mathit{Nat}$, and
$\_\dotdiv \_ : \mathit{Nat}\,\mathit{Nat} \rightarrow \mathit{Nat}$,
for the maximum, minimum and  ``monus'' (subtraction)
functions. The axioms $B_{\Delta}$ are the commutativity of the $\mathit{max}$
and $\mathit{min}$ functions.  The rules $R_{\Delta}$ for the defined
functions are: $\mathit{max}(N,N + M) \rightarrow N + M$,
$\mathit{min}(N,N + M) \rightarrow N$,
$N \dotdiv (N + M) \rightarrow 0$, and
$(N + M) \dotdiv N \rightarrow M$, where $N$ and $M$ have sort
$\mathit{Nat}$.

The predicates $\in$ and $\subset$ need not be explicitly defined, since they
can be expressed by the definitional equivalences $N \in NS = \mathit{tt}  \,
\Leftrightarrow N,NS = NS$, and
 $NS \subset NS' = \mathit{tt}  \,
\Leftrightarrow NS \subseteq NS'  = \mathit{tt} \, \wedge \, NS \not= NS'$.
\end{example}

FVP is a \emph{semi-decidable} property \cite{variants-of-variants}, which can be easily
verified (when it holds) by checking, using folding variant narrowing
(supported by Maude 2.7),
that for each function symbol $f: s_{1} \ldots s_{n} \rightarrow s$ the term $f(x_{1},\ldots,x_{n})$,
with $x_{i}$ of sort $s_{i}$, $1 \leq i \leq n$, has a finite
number of most general variants.  Given
an FVP decomposition $\mathcal{R}$ its \emph{variant complexity}
is the total number $n$ of variants for all such
$f(x_{1},\ldots,x_{n})$, provided
$f$ has some associated
rules of the form $f(t_1,\ldots,t_n) \rightarrow t'$.  This
gives a \emph{rough} measure of how costly it is to perform
variant computations \emph{relative} to the cost of performing
$B$-unification.  For example, the variant complexity
of $\mathit{NatSet}$ above is 20.

To be able to express 
systems of equations, say, $u_{1}=v_{1} \, \wedge \, \ldots \, \wedge \,  u_{n}=v_{n}$,
as \emph{terms}, we can extend  an MS signature $\Sigma$
with sorts $S$ to an OS signature $\Sigma^{\wedge}$
by: (1)
adding to $S$
  fresh new sorts 
$\mathit{Lit}$ and $\mathit{Conj}$
with a subsort inclusion
$\mathit{Lit} < \mathit{Conj}$; (2) adding a binary
conjunction operator $\_ \wedge\_ : \mathit{Lit} \; \mathit{Conj}
\rightarrow \mathit{Conj}$; and 
(3) adding for each $s \in S$  binary operators
 $\_=\_ : s \; s \rightarrow \mathit{Lit}$
and $\_ \not= \_ : s \; s \rightarrow \mathit{Lit}$.

Variant-based unification goes back to \cite{variant-JLAP}.
The paper \cite{var-sat}
gives a more precise characterization using $\Sigma^{\wedge}$-terms
as follows. If  $\mathcal{R}=(\Sigma,B,R)$ is an FVP decomposition of
$(\Sigma,E)$ and $B$ has a finitary $B$-unification
 algorithm, given a system of $\Sigma$-equations $\phi$ with variables
 $W$,
 folding variant narrowing
computes a \emph{finite} set $\mathit{VarUnif}^{W}_{E}(\phi)$
of $E$-unifiers away from $W$ that is \emph{complete} in the strong
sense that if $\alpha$ is an $R,B$-normalized
$E$-unifier of $\phi$ there exists $\theta
\in \mathit{VarUnif}^{W}_{E}(\phi)$ 
and an $R,B$-\emph{normalized}
$\rho$ such that  $\alpha|_{W} =_{B}
(\theta \rho)|_{W}$.

\vspace{-2ex}

\section{Constructor Variants and OS-Compactness}
\label{gvars-gvunif}

\vspace{-.5ex}

We gather some technical notions and results needed for the
inductive satisfiability procedure given in Section
\ref{inductive-validity-with-preds}.

%
The notion of \emph{constructor variant} answers
 the question: what variants of $t$ cover as instances
 modulo $B_{\Omega}$ all canonical forms of all ground instances of $t$?
The following lemma (stated and proved at the more general
order-sorted level in \cite{gutierrez-meseguer-var-pred-tech-rep}, but stated here for
the MS case for simplicity) gives a precise answer under
reasonable assumptions.  For more on constructor variants see 
\cite{var-sat,MLAVBS-WRLA16,gutierrez-meseguer-var-pred-tech-rep}.

\begin{lemma} \label{non-free-ctror-var-lemma}
 Let $\mathcal{R}=(\Sigma,B,R)$ be an FVP decomposition of $(\Sigma,E)$
protecting a constructor decomposition
$\mathcal{R}_{\Omega}=(\Omega,B_{\Omega},R_{\Omega})$.
Assume that: (i) $\Sigma =\Omega \cup \Delta$ with
$\Omega \cap \Delta = \emptyset$;
 (ii) $B$ has a
finitary $B$-unification algorithm and 
$B = B_{\Omega} \uplus B_{\Delta}$, with $B_{\Omega}$
$\Omega$-equations and
 if $u=v \in B_{\Delta}$,
$u,v$ are non-variable $\Delta$-terms.
%
%
Call $\llbracket t \rrbracket^{\Omega}_{R,B} = \{ (v,\theta) \in   \llbracket t
\rrbracket_{R,B} \mid  v \in T_{\Omega}(X)\}$ the set of
\emph{constructor variants} of $t$.
If $[u] \in \mathcal{C}_{\mathcal{R}_{\Omega}}$
is of the form $u =_{B} (t \gamma)!_{R,B}$, then
there is $(v,\theta) \in \llbracket t
\rrbracket^{\Omega}_{R,B}$ and a normalized ground substitution $\tau$
such that $u =_{B} v \tau$.  
%
%
\end{lemma}

We finally need the notion of an order-sorted
OS-\emph{compact} equational OS-FO theory $((\Sigma,\Pi),\Gamma)$,
generalizing  the compactness notion in
\cite{DBLP:journals/tcs/Comon93}.  
The notion is the \emph{same} (but called MS-compactness)
 for the special case of
MS theories treated in the preliminaries to simplify the exposition.
It is stated here in the more general OS case because
the satisfiability  algorithm in Section
\ref{inductive-validity-with-preds} works for the more general OS case,
and the paper's examples are in fact OS theories.

Given a OS equational theory
$(\Sigma,E)$, call a $\Sigma$-equality $u = v$
$E$-\emph{trivial} iff $u =_{E} v$, and a
$\Sigma$-disequality $u \not= v$
$E$-\emph{consistent} iff $u \not=_{E} v$.
Likewise, call a conjunction $\bigwedge D$
of $\Sigma$-disequalities $E$-\emph{consistent} iff each $u \not= v$
in $D$ is so. 
Call a sort $s \in S$ \emph{finite} in both $(\Sigma,E)$ and
$T_{\Sigma/E}$ iff
$T_{\Sigma/E,s}$ is a finite set, and \emph{infinite} otherwise.

\begin{definition} \label{compact-theo-defn} An equational OS-FO
  theory $((\Sigma,\Pi),\Gamma)$ is called \emph{OS-compact} iff: (i)
  for each sort $s$ in $\Sigma$ we can effectively determine whether
  $s$ is finite or infinite in
  $T_{\Sigma \cup \Pi/\widetilde{\Gamma},}$, and, if finite, can
  effectively compute a representative ground term
  $\mathit{rep}([u]) \in [u]$ for each
  $[u] \in T_{\Sigma \cup\Pi/\widetilde{\Gamma},s}$; (ii)
  $=_{\widetilde{\Gamma}}$ is decidable and $\widetilde{\Gamma}$ has a
  finitary unification algorithm; and (iii) any finite conjunction
  $\bigwedge D$ of negated $(\Sigma,\Pi)$-atoms whose variables all
  have infinite sorts and such that $\bigwedge \widetilde{D}$ is
  $\widetilde{\Gamma}$-consistent is satisfiable in
  $T_{\Sigma,\Pi,\Gamma}$.

Call an OS theory $(\Sigma,E)$
OS-\emph{compact} iff OS-FO theory  $((\Sigma,\emptyset),E)$ is OS-\emph{compact}.
\end{definition}

The key theorem, generalizing
a similar one in \cite{DBLP:journals/tcs/Comon93} 
is the following:

\begin{theorem} \label{compact-sat} \cite{var-sat,var-sat-short}
If  $((\Sigma,\Pi),\Gamma)$
is an \emph{OS-compact} theory, then satisfiability of QF
$(\Sigma,\Pi)$-formulas in  $T_{\Sigma,\Pi,\Gamma}$ is decidable.
\end{theorem}

The following OS-compactness results are proved in detail in
\cite{var-sat}: (i) a free constructor decomposition modulo axioms
$\mathcal{R}_{\Omega}=(\Omega,B_{\Omega},\emptyset)$ for
$B_{\Omega}$ any combination of associativity, commutativity and
identity axioms, except associativity without commutativity, is
OS-compact; and (ii) the constructor decompositions for
 parameterized modules for lists, compact lists,
multisets, sets, and hereditarily finite (HF) sets are all
\emph{OS-compact-preserving}, in the sense that if the actual parameter has an 
OS-compact constructor decomposition, then the corresponding instantiation
of the parameterized constructor decomposition is OS-compact.

\begin{example} \label{nat-set-ctor-OS-compact}
The constructor decomposition
  $\mathcal{R}_{\Omega_{c}}=(\Omega,B_{\Omega_{c}},R_{\Omega_{c}})$
  for the $\mathit{NatSet}$ theory in Example
  \ref{nat-set-fvp-example} is OS-compact.  This follows from the fact
  that $\mathit{NatSet}$ with set containment predicate
  $\_ \subseteq \_$ is just the instantiation of the constructor
  decomposition for the parameterized module of (finite) sets in \cite{var-sat} to the
  natural numbers with $0$, $1$, and $\_+\_$, which is itself a theory of
  free constructors modulo associativity, commutativity and identity
  $0$  for $\_+\_$ and therefore OS-compact by (i), so that,  by (ii),
  $\mathcal{R}_{\Omega_{c}}=(\Omega,B_{\Omega_{c}},R_{\Omega_{c}})$ is also OS-compact.
\end{example}

\vspace{-4ex}

\section{QF Satisfiability in Initial Algebras with Predicates} 
\label{inductive-validity-with-preds}

\vspace{-.5ex}

The known variant-based quantifier-free (QF) satisfiability and validity
results \cite{var-sat,var-sat-short} apply to the initial algebra $T_{\Sigma/E}$
of an equational theory $(\Sigma,E)$ having an FVP
variant-decomposition $\mathcal{R}=(\Sigma,B,R)$ protecting a constructor decomposition
  $\mathcal{R}_{\Omega}=(\Omega,B_{\Omega},R_{\Omega})$ and such that:
(i) $B$ has a finitary unification algorithm; and (ii) the equational
theory of $\mathcal{R}_{\Omega}=(\Omega,B_{\Omega},R_{\Omega})$ is
OS-compact.

\begin{example} QF validity and satisfiability in the initial algebra
  $T_{\Sigma/E}$ for $(\Sigma,E)$ the theory with the $\mathit{NatSet}$  FVP 
variant-decomposition $\mathcal{R}=(\Sigma,B,R)$ in Example \ref{nat-set-fvp-example}
are decidable because its axioms $B$ have a finitary unification
algorithm and, as explained in Example \ref{nat-set-ctor-OS-compact},
its constructor decomposition
$\mathcal{R}_{\Omega}=(\Omega,B_{\Omega},R_{\Omega})$
is OS-compact.
\end{example}

The decidable inductive validity and satisfiability results in
\cite{var-sat,var-sat-short} apply indeed to many \emph{data
  structures} of interest, which may obey structural axioms $B$ such
as commutativity, associativity-commutativity, or identity.  Many useful
examples are given in \cite{var-sat}, and a prototype Maude
implementation is presented in \cite{MLAVBS-WRLA16}.  There is,
however, a main limitation about the range of examples to which these
results apply, which this work directly addresses.  The limitation comes
from the introduction of  \emph{user-definable predicates}.  Recall
that we represent a predicate $p$ with sorts $s_{1},\ldots,s_{n}$
as a function $p:s_{1},\ldots,s_{n} \rightarrow \mathit{Pred}$
defined in the \emph{positive} case by confluent and terminating
 equations $p(u^{i}_{1}, \ldots,u^{i}_{n})=
\mathit{tt}$, $1 \leq i \leq k$.  The key problem with such predicates
$p$ is that, except in trivial cases, there are typically ground terms
$p(v_{1}, \ldots,v_{n})$ for which the predicate does \emph{not} hold.
This means that 
 $p$ must be a \emph{constructor} operator of sort $\mathit{Pred}$
which is \emph{not} a free constructor modulo the axioms $B_{\Omega}$.
This makes proving OS-compactness for a constructor
decomposition $\mathcal{R}_{\Omega}=(\Omega,B_{\Omega},R_{\Omega})$
including user-definable predicates a non-trivial case-by-case
task.  For example, the proofs of
OS-compactness for the set containment predicate
  $\_ \subseteq \_$ in the parameterized module of finite sets and for
other such predicates in other FVP parameterized modules
in \cite{var-sat} all required non-trivial analyses.
Furthermore, OS-compactness may fail for some 
$\mathcal{R}_{\Omega}$ precisely because of predicates (see Example \ref{nat-set-pred-example} below).

\begin{example} \label{nat-set-pred-example}
Consider the following extension by predicates
  $\mathit{NatSetPreds}$ of the $\mathit{NatSet}$ theory in Example
  \ref{nat-set-fvp-example}, where the constructor signature
  $\Omega = \Omega_{c} \uplus \Omega_{\Pi}$ adds the subsignature
  $\Omega_{\Pi}$ containing the strict order predicate
  $\_>\_: \mathit{Nat} \, \mathit{Nat} \rightarrow \mathit{Pred}$, the
  ``sort predicate''
  $\_\!\!:\!\!\mathit{Nat}: \mathit{NatSet} \rightarrow
  \mathit{Pred}$,
  characterizing when a set of natural numbers is a natural, and the
  even and odd predicates
  $\mathit{even}, \mathit{odd}: \mathit{NatSet} \rightarrow
  \mathit{Pred}$,
  defined by the rules $R_{\Pi}$:
  $N + M + 1 > N \rightarrow \mathit{tt}$,
  $N \!\!:\!\!\mathit{Nat} \rightarrow \mathit{tt}$,
  $\mathit{even}(N + N) \rightarrow \mathit{tt}$,
  $\mathit{odd}(N + N + 1) \rightarrow \mathit{tt}$, where $N$ and $M$
  have sort $\mathit{Nat}$.  $\mathit{NatSetPreds}$ is FVP, but its
  constructor decomposition
  $\mathcal{R}_{\Omega}=(\Omega_{c} \uplus
  \Omega_{\Pi},B_{\Omega_{c}},R_{\Omega_{c}} \uplus R_{\Pi})$
  is \emph{not} OS-compact, since the negation of the trichotomy law
  $N > M \vee M > N \vee N=M$ is the $B_{\Omega_{c}}$-consistent but
  \emph{unsatisfiable} conjunction of disequalities
  $N > M \not= \mathit{tt} \wedge M > N \not= \mathit{tt} \wedge N
  \not= M$.
\end{example}

The goal of this work is to provide a decision procedure for validity
and satisfiability of QF formulas in the initial algebra of an FVP
theory $\mathcal{R}$ that may contain user-definable predicates and
protects a constructor decomposition $\mathcal{R}_{\Omega}$ that need
not be OS-compact, under the following reasonable assumptions: (1)
$\mathcal{R} =(\Delta \uplus \Omega_{c} \uplus
\Omega_{\Pi},B_{\Delta} \uplus B_{\Omega_{c}},R_{\Delta} \uplus
R_{\Omega_{c}} \uplus R_{\Pi})$
protects
$\mathcal{R}_{\Omega}=(\Omega_{c} \uplus
\Omega_{\Pi},B_{\Omega_{c}},R_{\Omega_{c}} \uplus R_{\Pi})$,
where $\Omega_{\Pi}$ consists only of predicates, and $R_{\Pi}$
consists of rules of the form
$p(u^{i}_{1}, \ldots,u^{i}_{n}) \rightarrow \mathit{tt}$,
$1 \leq i \leq k_{p}$, defining each $p \in \Omega_{\Pi}$;
furthermore, $\mathcal{R}_{\Omega}$ satisfies conditions (i)--(ii) in Lemma
\ref{non-free-ctror-var-lemma};
(2) $\mathcal{R}_{\Omega_{c}}=(\Omega_{c},B_{\Omega_{c}},R_{\Omega_{c}})$
is OS-compact, its finite sorts (if any) are different from 
$\mathit{Pred}$, and is
the constructor decomposition of $(\Delta \uplus
\Omega_{c},B_{\Delta} \uplus B_{\Omega_{c}},R_{\Delta} \uplus
R_{\Omega_{c}})$;
and (3) each $p \in \Omega_{\Pi}$ has an associated set of
\emph{negative constrained patterns} of the form:
\[ \bigwedge_{1 \leq l \leq n_{j}} {w^{j}}_{l} \not= {w'^{j}}_{l} \,
\Rightarrow p({v^{j}}_{1}, \ldots,{v^{j}}_{n}) \not= \mathit{tt}, \;\; 1
\leq j \leq m_{p}
\]
with the ${v_{i}^{j}}$, ${w^{j}}_{l}$ and ${w'^{j}}_{l}$ 
$\Omega_{c}$-terms with variables in $Y_{j} =\mathit{vars}(p({v^{j}}_{1}, \ldots,{v^{j}}_{n}))$.
These negative constrained patterns are interpreted as meaning that
the following \emph{semantic equivalences} are valid in
$\mathcal{C}_{\mathcal{R}}$ for each $p \in \Omega_{\Pi}$, where
$\rho_{j} \in \{ \rho \in [Y_{j} \sra T_{\Omega_{c}}] \mid
\rho = \rho!_{R,B}\}$, $B=B_{\Delta} \uplus B_{\Omega_{c}}$,
and $R=R_{\Delta} \uplus
R_{\Omega_{c}} \uplus R_{\Pi}$:
\[[p({v^{j}}_{1}, \ldots,{v^{j}}_{n})\rho_{j}] \in
\mathcal{C}_{\mathcal{R}} \Leftrightarrow 
\bigwedge_{1 \leq
  l \leq n_{j}} ({w^{j}}_{l} \not= {w'^{j}}_{l}) \rho_{j}
\]
{\footnotesize
\[
 [p(t_{1}, \ldots,t_{n})] \in
\mathcal{C}_{\mathcal{R}} \Leftrightarrow \exists j \exists \rho_{j} \;
[p(t_{1}, \ldots,t_{n})] = [p({v^{j}}_{1}, \ldots,{v^{j}}_{n})\rho_{j}] \wedge \bigwedge_{1 \leq
  l \leq n_{j}} ({w^{j}}_{l} \not= {w'^{j}}_{l}) \rho_{j}
\]
}
The first equivalence means that any instance of a negative pattern by a normalized
ground substitution $\rho_{j}$ satisfying its
constrain is normalized,
so that $\mathcal{C}_{\mathcal{R}} \models  p({v^{j}}_{1}, \ldots,{v^{j}}_{n})\rho_{j}  \not= \mathit{tt}$.
The second means that $[p(t_{1}, \ldots,t_{n})] \in
C_{\mathcal{R}}$ iff $[p(t_{1}, \ldots,t_{n})]$ instantiates
a negative  pattern satisfying its constraint.

\begin{example}\label{nat-set-pred-patterns}  The module  $\mathit{NatSetPreds}$ from
Example \ref{nat-set-pred-example} satisfies above conditions
(1)--(3).  Indeed, (1), including conditions (i)--(ii) in Lemma
\ref{non-free-ctror-var-lemma}, follows easily from
its definition and that of $\mathit{NatSet}$, and (2)
also follows easily from the definition of  $\mathit{NatSet}$  and the remarks in
Example \ref{nat-set-ctor-OS-compact}. 
This leaves us with condition (3), where the negative constrained
patterns for 
$\Omega_{\Pi} =\{\_>\_,\mathit{even}, \mathit{odd},\_\!:\!\!\mathit{Nat} \}$ are the following:
\begin{itemize}
\item $N > N + M \not= \mathit{tt}$

\item $\mathit{even}(N + N + 1) \not= \mathit{tt}$, $\mathit{even}(\emptyset) \not= \mathit{tt}$,
 $(N \subseteq NS \not= \mathit{tt} \, \wedge \,  NS \not=
\emptyset) \Rightarrow \mathit{even}(N,NS) \not= \mathit{tt}$

\item $\mathit{odd}(N + N) \! \not= \! \mathit{tt}$,  $\mathit{odd}(\emptyset) \not= \mathit{tt}$,
$(N \! \subseteq \! NS \not= \mathit{tt} \wedge NS \not=
\emptyset) \Rightarrow \mathit{odd}(N,NS) \! \not= \! \mathit{tt}$

\item $\emptyset \!:\!\!\mathit{Nat} \not= \mathit{tt}$,
$(N \subseteq NS \not= \mathit{tt} \, \wedge \,  NS \not=
\emptyset) \Rightarrow (N,NS)\!:\!\!\mathit{Nat}  \not= \mathit{tt}$.
\end{itemize}

\noindent where $N$ and $M$ have sort $\mathit{Nat}$ and $NS$ sort $\mathit{Natset}$.
As explained in Appendix A of \cite{gutierrez-meseguer-var-pred-tech-rep}, the first
equivalence can be automatically checked using folding variant
narrowing. For a proof that the two equivalences hold in $\mathcal{C}_{\mathcal{R}}$ for
these predicates and their patterns (a few patterns are missing in the
proof  by mistake)
 see \cite{gutierrez-meseguer-var-pred-tech-rep}.
\end{example}


\noindent {\bf The Inductive Satisfiability Decision Procedure}.
Assume $\mathcal{R}$ satisfies conditions (1)--(3) above and let
$\Sigma = \Delta \uplus \Omega_{c} \uplus \Omega_{\Pi}$, and $E$ be
the axioms $B$ plus the equations associated with the rules $R$ in
$\mathcal{R}$.  Given a QF $\Sigma$-formula $\varphi$ the procedure
decides if $\varphi$ is satisfiable in $\mathcal{C}_{\mathcal{R}}$. We
can reduce the inductive validity decision problem of whether
$\mathcal{C}_{\mathcal{R}} \models \varphi$ to deciding whether
$\neg \varphi$ is unsatisfiable in $\mathcal{C}_{\mathcal{R}}$.  Since
any QF $\Sigma$-formula $\varphi$ can be put in disjunctive normal
form, a disjunction is satisfiable in $\mathcal{C}_{\mathcal{R}}$ iff
one of the disjuncts is, and all predicates have been turned into
functions of sort $\mathit{Pred}$, it is enough to decide the
satisfiability of a conjunction of $\Sigma$-literals of the form
$\bigwedge G \wedge \bigwedge D$, where the $G$ are equations and the
$D$ are disequations.  The procedure performs the following steps:

\begin{enumerate}
\item {\bf Unification}.  Satisfiability of the conjunction $\bigwedge
  G \wedge \bigwedge D$ is replaced by satisfiability for
some conjunction in the set 
$\{ (\bigwedge D \alpha)!_{R,B} \mid \alpha \in
\mathit{VarUnif}_{E}(\bigwedge G)\}$, discarding any obviously
unsatisfiable $(\bigwedge D \alpha)!_{R,B}$ in such a set.

\item $\Pi$-{\bf Elimination}.  After Step (1), each
  conjunction is a conjunction of disequalities $\bigwedge D'$.  If $\bigwedge D'$ is
a $\Delta \uplus \Omega_{c}$-formula, we go directly to Step (3);
otherwise $\bigwedge D'$ has the form
$\bigwedge D' = \bigwedge D_{1} \wedge p(t_{1}, \ldots,t_{n}) \not=
\mathit{tt} \wedge \bigwedge D_{2}$, where $p \in \Omega_{\Pi}$ and $D_{1}$ and/or $D_{2}$ may be 
empty conjunctions.  We then replace $\bigwedge D'$ by all
not obviously unsatisfiable conjunctions of the form:
\[
(\bigwedge D_{1} \wedge \bigwedge_{1 \leq l \leq n_{j}} {w^{j}}_{l} \not= {w^{',j}}_{l} \wedge
\bigwedge D_{2}) \theta \alpha
\]
where $1 \leq j \leq m_{p}$, $W = \mathit{vars}(\bigwedge D')$,
 $(p(t'_{1}, \ldots,t'_{n}),\theta) \in 
\llbracket p(t_{1}, \ldots,t_{n}) \rrbracket^{W,\Omega}_{R,B}$, and
$\alpha$ is a \emph{disjoint} $B_{\Omega_{c}}$-unifier of
the equation $p(t'_{1}, \ldots,t'_{n}) =
p({v^{j}}_{1},\ldots,{v^{j}}_{n})$ (i.e., sides are
renamed to \emph{share no variables} and 
$\mathit{ran}(\alpha) \cap (W \cup \mathit{ran}(\theta))=\emptyset$).
We use the negative constrained
patterns of $p$ and the constructor variants of $p(t_{1},
\ldots,t_{n})$ to \emph{eliminate} the disequality $ p(t_{1}, \ldots,t_{n}) \not=
\mathit{tt}$.  If  for some
$p' \in \Omega_{\Pi}$ some
disequality remains in 
$(\bigwedge D_{1} \wedge \bigwedge D_{2}) \theta \alpha$, we iterate Step 2.

\item  {\bf Computation of} $\Omega^{\wedge}_{c}$-{\bf Variants and}
 {\bf Elimination of Finite Sorts}. For $\bigwedge D'$ 
a $\Delta \uplus \Omega_{c}$-conjunction of disequalities, viewed
as a $(\Delta \uplus \Omega_{c})^{\wedge}$-term
its constructor $\Omega ^{\wedge}_{c}$-variants are of the
form $(\bigwedge D'',\gamma)$, with $\bigwedge D''$ an
$\Omega_{c}$-conjunction of disequalities.  
The variables of $\bigwedge D''$ are then
$Y_{\mathit{fin}} \uplus Y_{\infty}$, with $Y_{\mathit{fin}}$ the variables whose sorts are finite,
and $Y_{\infty}$ the variables with infinite sorts.
Compute all normalized ground substitution $\tau$ of the variables $Y_{\mathit{fin}}$
obtained by: (i) independently choosing for
each variable $y \in Y_{\mathit{fin}}$
a canonical representative for the sort of $y$
in all possible ways, and (ii) checking that for the $\tau$ so chosen $\bigwedge D'' \tau$ is
normalized, keeping $\tau$ if this holds and discarding it otherwise.
 Then $\bigwedge D'$ is satisfiable in
$\mathcal{C}_{\mathcal{R}}$ iff some $\bigwedge D'' \tau$ so
obtained is $B_{\Omega_{c}}$-consistent for
some $\Omega ^{\wedge}_{c}$-variant
$(\bigwedge D'',\gamma)$ of $\bigwedge D'$.
\end{enumerate}

\begin{example} We can illustrate the use of the above decision procedure by
  proving the validity of the QF formula
  $\mathit{odd}(N) = \mathit{tt} \Leftrightarrow \mathit{even}(N)
  \not= \mathit{tt}$
  in the initial algebra $\mathcal{C}_{\mathcal{R}}$ of
  $\mathit{NatSetPreds}$.  That is, we need to show that its negation
  $(\mathit{odd}(N) = \mathit{tt} \wedge \mathit{even}(N)
  =\mathit{tt}) \vee (\mathit{odd}(N) \not= \mathit{tt} \wedge
  \mathit{even}(N) \not=\mathit{tt})$
  is unsatisfiable in $\mathcal{C}_{\mathcal{R}}$.  Applying the {\bf
    Unification} step to the first disjunct
  $\mathit{odd}(N) = \mathit{tt} \wedge \mathit{even}(N) =\mathit{tt}$
  no variant unifiers are found, making this disjunct unsatisfiable.
  Applying the $\Pi$-{\bf Elimination} step to the first disequality
  in the second disjunct
  $\mathit{odd}(N) \not= \mathit{tt} \wedge \mathit{even}(N)
  \not=\mathit{tt}$,
  since the only constructor variant of $\mathit{odd}(N)$ different
  from $\mathit{tt}$ is the identity variant, and the only disjoint
  $B_{\Omega_{c}}$-unifier of $\mathit{odd}(N)$ with the negative patterns for
  $\mathit{odd}$ is $\{N \mapsto M + M\}$ for the (renamed)
 unconstrained negative pattern $\mathit{odd}(M + M) \not= \mathit{tt}$,
  we get the disequality
  $\mathit{even}(M + M) \not=\mathit{tt}$, whose normal form
  $\mathit{tt} \not=\mathit{tt}$ is unsatisfiable.
\end{example}

\begin{theorem} \label{correctness-var-sat-proc}
 For FVP $\mathcal{R} =(\Delta \uplus \Omega_{c} \uplus
\Omega_{\Pi},B_{\Delta} \uplus B_{\Omega_{c}},R_{\Delta} \uplus
R_{\Omega_{c}} \uplus R_{\Pi})$
protecting
$\mathcal{R}_{\Omega}=(\Omega_{c} \uplus
\Omega_{\Pi},B_{\Omega_{c}},R_{\Omega_{c}} \uplus R_{\Pi})$ and
satisfying above conditions (1)--(3),
the above procedure correctly decides the satisfiability
of a QF $\Sigma$-formula $\varphi$ in the canonical term algebra
$\mathcal{C}_{\mathcal{R}}$.
\end{theorem}

\noindent {\bf Sort Predicates for Recursive Data Structures}.
We can axiomatize many (non-circular) recursive data structures as the
elements of an initial algebra $T_{\Omega}$ on a many-sorted signature
of free constructors $\Omega$.  For example, lists 
can be so axiomatized with $\Omega$ consisting of just two sorts,
$\mathit{Elt}$, viewed as a parametric sort of list elements,
 and $\mathit{List}$,
a constant $\mathit{nil}$  of sort $\mathit{List}$,
and a ``cons'' constructor $\_;\_ : \mathit{Elt} \, \mathit{List}
\rightarrow \mathit{List}$.    

In general, however, adding to such data structures defined
functions corresponding to ``selectors''  that can extract the
constituent parts of each data structure cannot be done in a
satisfactory way if we remain within a many-sorted setting.  For
example, for lists  we would like to have selectors
$\mathit{head}$ and $\mathit{tail}$ (the usual $\mathit{car}$ and
$\mathit{cdr}$ in Lisp notation).  
For $\mathit{head}$ the natural equation is
$\mathit{head}(x ; l)=x$. Likewise, the natural equation for
$\mathit{tail}$ is $\mathit{tail}(x ; l)=l$.
 But this leaves open the problem
of how to define $\mathit{head}(\mathit{nil})$, for which no
satisfactory solution exists.  J. Meseguer and J.A. Goguen proposed a simple
solution to this ``constructor-selector'' problem using initial order-sorted
algebras in \cite{csel}.  The key idea is the following.  For each
non-constant constructor symbol, say $c: A_{1} \ldots A_{n}
\rightarrow B$, $n \geq 1$,
we introduce a subsort $B_{c} < B$ and give the tighter typing
$c: A_{1} \ldots A_{n} \rightarrow B_{c}$.  The selector problem is now
easily solved by associating to each non-constant constructor $c$ selector functions
$\mathit{sel}^{c}_{i} :B_{c} \rightarrow A_{i}$, $1 \leq i \leq n$,
defined by the equations $\mathit{sel}^{c}_{i}(c(x_{1}, \ldots,
x_{n}))=x_{i}$, $1 \leq i \leq n$.  Outside the subsort $B_{c}$ the
selectors $\mathit{sel}^{c}_{i}$ are actually undefined.
For the above example of lists this just means adding a subsort
$\mathit{List}_{\_;\_} < \mathit{List}$, where
$\mathit{List}_{\_;\_}$ is usually written as $\mathit{NeList}$
(non-empty lists), and tightening  the typing of ``cons'' to
$\_;\_ : \mathit{Elt} \, \mathit{List}
\rightarrow \mathit{NeList}$.   In this way the
$\mathit{head}$ and $\mathit{tail}$ selectors have typings
$\mathit{head} : \mathit{NeList}
\rightarrow \mathit{Elt}$ and $\mathit{tail} : \mathit{NeList}
\rightarrow \mathit{List}$, again with equations 
$\mathit{head}(x ; l)=x$ and $\mathit{tail}(x ; l)=l$, with 
$x$ of sort $\mathit{Elt}$ and $l$ of sort $\mathit{List}$.

We have just described a  general theory transformation
$\Omega \mapsto (\widetilde{\Omega} \uplus \Delta,E_{\Delta})$
from any MS signature $\Omega$ to an OS theory with selectors $\Delta$.
Due to space limitations, the following key facts
are discussed in detail in Section 4.2 of \cite{gutierrez-meseguer-var-pred-tech-rep}:
(1) $(\widetilde{\Omega} \uplus \Delta,\emptyset,R(E_{\Delta}))$ is FVP
with $(\widetilde{\Omega},\emptyset,\emptyset)$  as its constructor decomposition.
(2) To increase expressiveness, we can define for each subsort $B_{c}$ associated with a constructor
$c$ a corresponding equationally-defined sort predicate $\_\!
:\!\! B_{c}$, thus obtaining a decomposition 
$(\widetilde{\Omega} \uplus \Pi  \uplus \Delta, \emptyset,R(E_{\Delta}) \uplus
R(E_{\Pi}))$ that is also FVP.  (3) Each sort predicate $\_\!
:\!\! B_{c}$ has an associated set of \emph{negative patterns}, so
that our variant satisfiability algorithm makes satisfiability
of QF formulas in the initial algebra $T_{\widetilde{\Omega} \uplus \Pi
  \uplus\Delta/E_{\Delta} \uplus E_{\Pi}}$ \emph{decidable}.

\begin{example} (Lists of Naturals with Sort Predicates).
We can instantiate the above order-sorted theory of lists with 
selectors $\mathit{head}$ and $\mathit{tail}$ by 
instantiating the  parameter sort $\mathit{Elt}$  to
a sort $\mathit{Nat}$ with constant $0$, subsort
$\mathit{NzNat} < \mathit{Nat}$, and unary
constructor $s: \mathit{Nat} \rightarrow \mathit{NzNat}$
with  selector $p: \mathit{NzNat} \rightarrow \mathit{Nat}$
satisfying the equation $p(s(n))=n$.  We then extend this
specification with sort predicates 
$\_\! :\!\! \mathit{NzNat} : \mathit{Nat} \rightarrow \mathit{Pred}$
and
$\_\! :\!\! \mathit{NeList} : \mathit{List} \rightarrow
\mathit{Pred}$,
defined by equations
$n' \! :\!\! \mathit{NzNat} = \mathit{tt}$
and $l' \! :\!\! \mathit{NeList} = \mathit{tt}$,
with $n'$ of sort $\mathit{NzNat}$ and $l'$ of sort $\mathit{NeList}$.
Their corresponding negative patterns are:
$0 \! :\!\! \mathit{NzNat} \not= \mathit{tt}$ and
$\mathit{nil} \! :\!\! \mathit{NeList} \not= \mathit{tt}$.

One advantage of adding these sort predicates is that some properties not
expressible as QF formulas become QF-expressible.  For example, to
state that every number is either $0$ or a non-zero number
(resp. every list is either $\mathit{nil}$ or a non-empty list) we need
the formula $n=0 \vee (\exists n') \; n=n'$ (resp.
$l= \mathit{nil} \vee (\exists l') \; l=l'$), where $n$ has sort
$\mathit{Nat}$ and
$n'$ sort $\mathit{NzNat}$ (resp. $l$ has sort
$\mathit{List}$ and
$l'$ sort $\mathit{NeList}$).  But with sort predicates this
can be expressed by means of the QF formula
$n=0 \vee  n \! :\!\! \mathit{NzNat} = \mathit{tt}$ 
(resp. $l= \mathit{nil} \vee l \! :\!\! \mathit{NeList} = \mathit{tt}$).
\end{example}

\vspace{-5ex}

\section{Implementation} \label{implementation}

\vspace{-.5ex}

We have implemented the variant satisfiability decision procedure
of  Section~\ref{inductive-validity-with-preds} in a
new prototype tool. The implementation consists of 11 new Maude modules (from
17 in total), 2345 new lines of code, and uses the Maude's
\texttt{META-LEVEL} to carry out the steps of the
procedure in a reflective way. We have also developed a Maude interface to ease
the definition of properties and patterns as equations.
The three steps of the variant satisfiability procedure are
implemented using Maude's \texttt{META-LEVEL} functions.  Let us illustrate them
for $\mathit{NatSetPreds}$.

\begin{example} \label{tool-example}
We can prove the inductive validity of the formula
$\texttt{N - M} = \texttt{0} \Leftrightarrow (\texttt{M > N} =
\texttt{tt} \vee \texttt{N} = \texttt{M})$,
where $\texttt{N - M}$ denotes $\texttt{N}$ ``monus'' $\texttt{M}$,
by showing that each conjunction in its negation,
$(\texttt{N - M} = \texttt{0} \wedge \texttt{M > N} \neq
    \texttt{tt} \wedge \texttt{N} \neq \texttt{M}) \vee (\texttt{N -
      M} \neq \texttt{0} \wedge \texttt{M > N} = \texttt{tt}) \vee
    (\texttt{N - M} \neq \texttt{0} \wedge \texttt{N} = \texttt{M})$
    is unsatisfiable.  For the first conjunct the algorithm's three
    steps are as follows.
 After the {\bf unification} step, we obtain
 $\texttt{(V2 + V3) > V2} \neq
      \texttt{tt} \wedge \texttt{V2} \neq \texttt{V2 + V3}$, where
    \texttt{V2} and \texttt{V3} are variables of sort
    \texttt{Natural}.
  Applying the $\Pi$-{\bf elimination} step, we obtain:
    $\texttt{V4} \neq \texttt{V4 + 0}$, where
    \texttt{V4} is a variable of sort \texttt{Natural}. After
    normalization, the formula becomes $B_{\Omega_{c}}$-inconsistent
    and therefore unsatisfiable.  The other two conjuncts are likewise unsatisfiable.
 \end{example}

For a  more detailed discussion of the
implementation see Section 5 of \cite{gutierrez-meseguer-var-pred-tech-rep}.

\section{Related Work and Conclusions} \label{related-concl}

The original paper proposing the concepts of variant and FVP is
\cite{comon-delaune}.  FVP ideas have been further advanced in
\cite{variant-JLAP,ciobaca-thesis,dran-etal-fvp,variants-of-variants}.
Variant satisfiability has been studied on
\cite{var-sat,var-sat-short,MLAVBS-WRLA16}.
In relation to that work, the main contribution of this
paper is the extension of variant satisfiability  to handle
user-definable predicates.

As mentioned in the Introduction, satisfiability decision procedures
can be either theory-specific or theory-generic.  Two recent advanced
textbooks on theory-specific decision procedures are \cite{DBLP:books/daglib/0019162}
and \cite{DBLP:series/txtcs/KroeningS08}.  These two classes of
procedures complement each other: theory specific ones are more
efficient; but theory-generic ones are user-definable and can
substantially increase the range of SMT solvers.

Other theory-generic
satisfiability approaches include: (i) the superposition-based one,
e.g., \cite{DBLP:conf/lics/LynchM02,DBLP:journals/iandc/ArmandoRR03,DBLP:conf/cade/LynchT07,DBLP:journals/tocl/ArmandoBRS09,DBLP:journals/scp/TushkanovaGRK15},
where it is proved that a
superposition theorem proving inference system terminates for a given first-order
theory together with any given set of ground clauses representing a
satisfiability problem; and (ii) that of
decidable theories defined by means of formulas with
triggers \cite{DBLP:journals/jar/DrossCKP16}, that
allows a user to define a new theory with decidable QF satisfiability 
by axiomatizing it according to some
requirements, and then making an SMT solver extensible by
such a user-defined theory. While not directly
comparable to the present one, these approaches
(discussed in more detail in \cite{var-sat}) can be seen as complementary
ones, further enlarging the repertoire of theory-generic satisfiability methods.

In conclusion, the present work has extended variant satisfiability to
support initial algebras specified by FVP theories with
user-definable predicates under fairly general conditions.  Since
such predicates are often needed in specifications, this substantially
enlarges the scope of variant-based initial
satisfiability algorithms.  The most obvious next step is to 
combine the original variant satisfiability  algorithm defined in \cite{var-sat,var-sat-short}
and implemented in \cite{MLAVBS-WRLA16} with the present one.
To simplify both the exposition and the prototype implementation,
a few simplifying assumptions, such as the assumption that the signature
$\Omega$ of constructors and that $\Delta$ of defined functions
share no subsort-overloaded symbols, have been made.  
For both greater
efficiency and wider applicability,
the combined generic
algorithm will drop such assumptions and will use constructor
unification \cite{var-sat,MLAVBS-WRLA16}. 

\vspace{2ex}

\bibliographystyle{splncs03}
\bibliography{tex}

\end{document}